\documentclass[11pt]{article}

\usepackage{mathpazo}
\usepackage{amsfonts}
\usepackage{amsmath}
\usepackage{graphicx}
\usepackage{latexsym}
\usepackage{mathabx}
\usepackage{mathrsfs}
\usepackage[margin=1.05in]{geometry}
  
\usepackage[T1]{fontenc}
\usepackage{times}
\usepackage{color,graphicx}
\usepackage{array}
\usepackage{enumerate}
\usepackage{amsmath}
\usepackage{amssymb}
\usepackage{amsthm}
\usepackage{pgfplots}
\usepackage{pgf}
\usepackage{tikz}
\usetikzlibrary{patterns}
\usepgfplotslibrary{patchplots} % LATEX and plain TEX
\usetikzlibrary{pgfplots.patchplots} % LATEX and plain TEX
\pgfplotsset{width=9cm,compat=1.5.1}

\usepackage{bbm, commath}
\usepackage{multicol}
\usepackage{tikz-network}
\usetikzlibrary {positioning,patterns, calc}
\usetikzlibrary{graphs,graphs.standard,quotes}
\usepackage{authblk}
\usepackage{enumitem}

\usepackage{amsthm}
\newtheorem{thm}{Theorem}

\newtheorem{prop}{Proposition}
\newtheorem{lem}{Lemma}
\newtheorem{cor}{Corollary}
\newtheorem{exm}{Example}
\newtheorem{con}{Conjecture}
\newtheorem{rem}{Remark}

\def \Nl {{\mathbbm N}}
\def \Zl {{\mathbbm Z}}
\def \Ql {{\mathbbm Q}}
\def \Rl {{\mathbbm R}}
\def \Cl {{\mathbbm C}}
\def \e {{\bf e}}

\def \x {{\bf x}}

\newcommand{\ob}[1]{\left(#1\right)}
\newcommand{\cb}[1]{\left\lbrace #1\right\rbrace}
\newcommand{\tb}[1]{\left[#1\right]}

\newcommand{\up}[1]{\overset{#1}{\uplus}~}

\newcommand{\eig}[1]{\sigma\left(#1\right)}

\newcommand{\be}{\mathbf e}

\usepackage{xcolor}
\usepackage{makeidx}
\usepackage[colorlinks=true,linkcolor=blue,anchorcolor=blue,citecolor=red,urlcolor=magenta]{hyperref}
\usepackage[alphabetic,backrefs]{amsrefs}

\begin{document}
\title{Quantum walks on blow-up graphs}

\author{
	Bikash Bhattacharjya,\textsuperscript{\!\!1} Hermie Monterde,\textsuperscript{\!\!2} Hiranmoy Pal\ \textsuperscript{\!\!3}
}

\maketitle

\begin{abstract}
A blow-up of $n$ copies of a graph $G$ is the graph $\up{n}G$ obtained by replacing every vertex of $G$ by an independent set of size $n$, where the copies of two vertices in $G$ are adjacent in the blow-up if and only if the two vertices are adjacent in $G$. Our goal is to investigate quantum state transfer on a blow-up graph $\up{n}G$ relative to the adjacency matrix. We characterize strong cospectrality, periodicity, perfect state transfer (PST) and pretty good state transfer (PGST) in blow-up graphs. It turns out, if $\up{n}G$ admits PST or PGST, then one must have $n=2.$ In particular, if $G$ has an invertible adjacency matrix, then each vertex in $\up{2}G$ pairs up with a unique vertex to exhibit strong cospectrality. We then apply our results to determine infinite families of graphs whose blow-ups admit PST and PGST.\\


\end{abstract}

%%%%%%%%%%% Introduction %%%%%%%%%%%%%%%%%%%%%%%%%%%%%%%%%%%%

\section{Introduction}
A quantum spin network is modelled by an undirected graph $G$, where the vertices and edges of $G$ represent the qubits and their interactions in the network. A \textit{continuous-time quantum walk} on $G$ is described by the unitary matrix
\begin{center}
    $U(t)=\exp{\ob{itA}}$
\end{center}
where $t\in\Rl$, $i=\sqrt{-1}$, and $A$ is the adjacency matrix of $G$. Here, $|U(t)_{u,v}|^2$ represents the probability that the quantum state assigned at vertex $u$ is found in vertex $v$ at time $t$. We say \textit{perfect state transfer} (PST) between $u$ and $v$ in $G$ at time $\tau>0$ if $|U\ob{\tau}_{u,v}|=1$.
In this case, the quantum state initially at vertex $u$ is transmitted to vertex $v$ at time $\tau$ with probability one. If $|U\ob{\tau}_{u,u}|=1$, then $u$ is said to be \textit{periodic} in $G$ at time $\tau$. A graph is \textit{periodic} if it is periodic at all vertices at the same time.

PST is critical for short distance communication in a physical quantum computing scheme \cite{JOHNSTON2017375}. The study of PST was initiated by Bose \cite{bose} and Christandl et al. \cite{chr1}. In \cite{chr2}, we find that $P_2$ and $P_3$ are the only paths that admit PST. Following this negative result, several classes of graphs admitting PST have been studied, like cubelike graphs \cite{ber, che}, distance-regular graphs \cite{cou2}, integral circulant graphs \cite{mil4, saxe}, Hadamard diagonalizable graphs \cite{JOHNSTON2017375,mclaren2023weak}, and quotient graphs \cite{bac,ge}. In \cite{chr2}, Christandl et al.\ used Cartesian products to construct graphs with PST over long distances. PST was also studied under joins \cite{ange1,kirkland2023quantum}, non-complete extended p-sums \cite{pal1,pal2} and other graph products \cite{ack1,ge,cou}.
\par

Since PST in simple unweighted graphs is rare \cite{god2}, the notion of \textit{pretty good state transfer} (PGST) was introduced \cite{god1,vin}. A graph $G$ exhibits PGST between $u$ and $v$ if there is a sequence $\tau_k\in\Rl$ such that $\lim_{k\rightarrow \infty}|U\ob{\tau_k}_{u,v}|=1$, i.e., $|U(t)_{u,v}|^2$ can be made arbitrarily close to one through appropriate choices of $t$. If $\lim_{k\rightarrow \infty}|U\ob{\tau_k}_{u,u}|=1$, then $G$ is said to be \textit{almost periodic at $u$} with respect to the sequence $\tau_k$, and $G$ is said to be \textit{almost periodic} if there is a sequence $\tau_k\in\Rl$ such that
\[\lim\limits_{k\to\infty}U\ob{\tau_k}=\gamma I,~\text{for some }\gamma\in\Cl,\]
where $I$ is the identity matrix of appropriate order. %Several classes of trees have been investigated for PGST, such as the 
Infinite families of paths \cite{cou3, god4, bom}, double stars \cite{fan} and circulant graphs are known to exhibit PGST \cite{pal6,pal7,pal4}. 

Our goal in this paper is to systematically study quantum state transfer on blow-up graphs. Here, we consider the more familiar type of blow-up where every vertex of a graph is replaced by an independent set. This variant of the blow-up operation is a special case of the lexicographic product. The \textit{lexicographic product} $G[H]$ of $G$ and $H$ is the graph $G[H]$ obtained by replacing every vertex of $G$ by a copy of $H$, and adding all possible edges between the vertices in the copies of $H$ corresponding to adjacent vertices in $G$. In \cite{ge}, Ge at al.\ provided sufficient conditions for PST to occur in lexicographic products. Their results about PST in $G[H]$ require that $H$ exhibits PST. But in our case, $H$ is an empty graph which cannot exhibit PST, and so their results are not applicable. This is our main motivation for studying quantum state transfer on blow-up graphs. Another motivation is to explore a graph operation that induces strong cospectrality between many pairs of vertices. Strong cospectrality is an important property because it is a requirement PGST. It turns out that if $0$ is not in the eigenvalue support of vertex $u$ in $G$, then the two copies of $u$ are strongly cospectral in $\up{2}G$ (see Theorem \ref{sc}(2)). Hence, each vertex in $\up{2}G$ pair up to exhibit strong cospectrality whenever $G$ has an invertible adjacency matrix (see Corollary \ref{sccor}). Thus, blow-ups have the advantage of producing pairs of vertices that are promising for high probability state transfer. Lastly, the two copies of a vertex in $G$ are twins in $\up{2}G$. This allows us to utilize some of the recent results on quantum state transfer between twins  \cite{kirk2,mon1,fr,pal8}. In particular, we give characterizations of strong cospectrality, PST, PGST and periodicity in blow-up graphs (see Theorems \ref{sc}-\ref{per}). We then characterize PST and PGST in blow-ups of special classes of graphs (see Sections \ref{qstgs}-\ref{cones}). We note that while  PST in $\up{2}G$ only occurs between vertices at distance two, one may use the Cartesian product operation to obtain relatively sparse graphs that admit PST at large distances (Section \ref{cons}).

\section{Blow-up}
Throughout, we let $G$ be a simple (loopless) connected undirected graph with vertex set $V$. The \textit{blow-up} of $G$, denoted by $\up{n}G$, is the graph with vertex set $\Zl_n\times V$, and two vertices $(l,u)$ and $(m,v)$ are adjacent in $\up{n}G$ if and only if the vertices $u$ and $v$ are adjacent in $G$. It is immediate that if $G$ is an $r$-regular graph then $\up{n}G$ is $nr$-regular. If $P_3$ is the path on three vertices then $\up{2}P_3$ is the graph shown in Figure \ref{fg1}.

A pair of vertices $u$ and $v$ in a graph $G$ are called twins if $N\ob{u}\setminus\cb{v}=N\ob{v}\setminus\cb{u}$. A twin set is a set of vertices in a graph whose elements are pairwise twins. If $u$ is a vertex of $G$, then $T_u=\cb{(j,u):j\in\Zl_n}$ is a set of twins in $\up{n}G.$ If $n\geq 2$, then each vertex in in $\up{n}G$ is contained in a twin set of size at least $n$.

\begin{prop}
\label{twins}
Let $u$ and $v$ be two vertices in $G$. The set $T_u\cup T_v$ is a twin set in $\up{n}G$ if and only if $u$ and $v$ are non-adjacent twins in $G.$
\end{prop} 

We also state Kronecker's approximation theorem which is useful in classifying graphs with PGST. 

\begin{thm}\cite{apo}\label{kron}
If $\alpha_1,\ldots,\alpha_l$ are arbitrary real numbers and if $1,\theta_1,\ldots, \theta_l$ are real, algebraic numbers linearly independent over $\Ql$ then for $\epsilon>0$ there exist $q\in\Zl$ and $p_1,\ldots,p_l\in\Zl$ such that
\[\left|q\theta_j-p_j-\alpha_j\right|<\epsilon.\]
\end{thm}

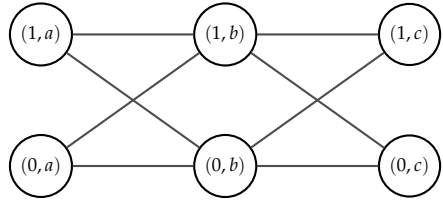
\begin{figure}[!ht]
\centering
\begin{tikzpicture}[scale=.35,auto=left, rotate=0]
\tikzstyle{every node}=[circle, thick, fill=white, scale=0.65]

 \node[draw] (1) at (0,0)  {$(0,a)$};
 \node[draw] (2) at (7,0) {$(0,b)$};
 \node[draw] (3) at (14,0) {$(0,c)$};

 \node[draw] (4) at (0,5) {$(1,a)$};
 \node[draw] (5) at (7,5) {$(1,b)$};
 \node[draw] (6) at (14,5) {$(1,c)$};
 
  \draw[thick,black!70] (1)--(2)--(3);
  \draw[thick,black!70] (4)--(5)--(6);
  \draw[thick,black!70] (1)--(5)--(3);
  \draw[thick,black!70] (4)--(2)--(6);

 \end{tikzpicture}
\caption{The graph $\up{2}P_3$, a blow-up of two copies of $P_3$}
\label{fg1}
\end{figure}

%%%%%%%%%%%%%%%%%%%%%%%%%%%%%%%%%%%%%%%%%%%%%%%%%%%%%%%%%%%%%
\section{Transition matrix}\label{qst}

Let $G$ be a graph on $m$ vertices with adjacency matrix $A.$ We refer to the eigenvalues of $A$ as the the eigenvalues of $G.$ Suppose $J_n$ is the square matrix of order $n\times n$ with all entries equal to $1$. Considering the lexicographic ordering on $\Zl_n\times V(G),$ the adjacency matrix of $\up{n}G$ is evaluated as $A_n=J_n\otimes A,$ which is the tensor product of the matrices $J_n$ and $A$. Let $\textbf{1}$ be the vector of length $n$ with all entries equal to $1$, and let $\lambda$ be an eigenvalue of $A$ associated to the eigenvector $\x.$ Then
\[\left(J_n\otimes A\right)\left(\textbf{1}\otimes \x\right)=J_n\textbf{1}\otimes A\x=n\lambda\left(\textbf{1}\otimes \x\right).\]

Let $\sigma(G)=\{\lambda_1,\lambda_2,\ldots,\lambda_d\}$ denote the distinct eigenvalues of $A$, and $E_j$ be the idempotent associated with $\lambda_j$. The \textit{eigenvalue support} of a vertex $u$ in $G$ is the set
\[\sigma_u(G)=\cb{\lambda\in\sigma(G): E_{\lambda}\e_{u}\neq \mathbf{0}}.\]
In case $0\in\eig{G},$ we assume $\lambda_d=0$. Then the idempotents of $A_n$  corresponding to the eigenvalues $n\lambda_1,n\lambda_2,\ldots,n\lambda_{d-1}$ and $0$ are respectively given by
\begin{eqnarray}
\label{sc1}
\frac{1}{n}J_n\otimes E_1,\ \ \frac{1}{n}J_n\otimes E_2,\ \ \ldots,\ \ \frac{1}{n}J_n\otimes E_{d-1},\ \ I_{mn}-\frac{1}{n}J_n\otimes\left(I_m-E_d\right).
\end{eqnarray}
Consequently, the spectral decomposition of the transition matrix of $\up{n}G$ becomes
\begin{eqnarray}\label{peq3}
    U_n(t)&=&\exp{\left(-itA_n\right)}=\sum\limits_{j=1}^{d-1}\exp{\left(-in\lambda_j t\right)}\cdot\frac{1}{n}J_n\otimes E_j+\left[I_{mn}-\frac{1}{n}J_n\otimes\left(I_m-E_d\right)\right].
\end{eqnarray}
For $(0,u),(1,u)\in \Zl_n\times V(G),$ one can then check that
\begin{eqnarray}
\hspace{-0.5in}
\e^T_{(0,u)}U_n(t)\e_{(1,u)} &=
& \frac{1}{n}\sum\limits_{j=1}^{d}\tb{\exp{\left(-in\lambda_j t\right)}-1} (E_j)_{u,u}.\label{eq3}
\end{eqnarray}  
If $0$ not an eigenvalue of $G,$ then we get the same set of idempotents for $A_n$ in (\ref{sc1}), except for the $d$th matrix which is now given by $I_{mn}-\frac{1}{n}J_n\otimes I_m$. Thus, the spectral decomposition of the transition matrix of $\up{n}G$ can be evaluated as follows
\begin{center}
    $U_n(t) = \exp{\left(-itA_n\right)}=\sum\limits_{j=1}^{d}\exp{\left(-in\lambda_j t\right)}\cdot\frac{1}{n}J_n\otimes E_j+\left[I_{mn}-\frac{1}{n}J_n\otimes I_m\right].$
\end{center}
Thus, we get an expression for $\e^T_{(0,u)}U_n(t)\e_{(1,u)}$ similar to \eqref{eq3}. Next, we state an observation about the spectra and eigenvalue supports in blow-ups (see also \cite{oliveira2014spectra}).

\begin{prop}\label{pp1}

Let $n\geq 2$. Then $\eig{\up{n}G}=n\cdot\eig{G}\cup \cb{0}$ and $\sigma_{(j,v)}\left(\up{n}G\right)=n\cdot\sigma_v(G)\cup\cb{0}$ for any vertex $u$ of $G$ and for all $j\in\Zl_n$.
\end{prop}

In \cite[Lemma 3.3]{cou}, the authors provided the form of transition matrices of graphs whose adjacency matrices look like $B\otimes M+C\otimes N$, where $(B,C)$ and $(M,N)$ are pairs of commuting matrices. This include blow-up graphs, since we may write $J_n\otimes A=A(K_n)\otimes A+I_n\otimes A$. However, their paper had no other results specific to blow-up graphs.

We also mention that Monterde investigated sedentariness, a type of low probability quantum transport, on several variants of the blow-up operation \cite{sed}. Apart from the work of Ge et al.\ \cite{ge} and Monterde, we are unaware of other results about quantum state transfer on blow-up graphs.

\section{Strong cospectrality}
Two vertices $u$ and $v$ in a graph $G$ are said to be \textit{strongly cospectral} if for every $\lambda\in\sigma_u(G)$,
\begin{equation*}
E_{\lambda}\textbf{e}_u=\pm E_{\lambda}\textbf{e}_v.
\end{equation*} 
Strong cospectrality is a necessary condition for PGST to occur between two vertices \cite[Lemma 13.1]{god1}. This motivates us to characterize strong cospectrality in blow-up graphs.

\begin{thm}
\label{sc}
Let $G$ be a graph with vertex $u$. The following hold.
\begin{enumerate}
\item If $n\geq 3$, then $\up{n}G$ does not exhibit strong cospectrality.
\item Let $n=2$. Then $(0,u)$ and $(1,u)$ are strongly cospectral in $\up{2}G$ if and only if $0\notin\sigma_u(G)$. Moreover, $(0,u)$ can only be strongly cospectral with $(1,u)$ in $\up{2}G$. 
\end{enumerate}
\end{thm}

\begin{proof}
We prove 1. If $n\geq 3$, then $|T_u|\geq 3$ in $\up{n}G$ for every vertex $u$ of $G$. Invoking \cite[Corollary 3.10]{mon1} yields the desired conclusion. To prove 2, note that the $d-1$ idempotents of $A_2$ in (\ref{sc1}) satisfy $\left(\frac{1}{2}J_2\otimes E_j\right)(\textbf{e}_0\otimes \textbf{e}_u)=\frac{1}{2}\textbf{1} \otimes E_j \textbf{e}_u=\left(\frac{1}{n}J_2\otimes E_j\right)(\textbf{e}_1\otimes \textbf{e}_u).$
Let $E$ denoted the $d$th idempotent of $A_2$ in (\ref{sc1}). Note that $E(\textbf{e}_0\otimes \textbf{e}_u)=\left(\textbf{e}_0-\frac{1}{2}\textbf{1}\right)\otimes \textbf{e}_u+\left(\frac{1}{2}\textbf{1}\otimes E_d\textbf{e}_u\right)$, and so we have $E(\textbf{e}_0\otimes \textbf{e}_u)\neq E(\textbf{e}_1\otimes \textbf{e}_u)$. However, one checks that  $E(\textbf{e}_0\otimes \textbf{e}_u)=- E(\textbf{e}_1\otimes \textbf{e}_u)$ if and only if $n=2$ and $0\notin\sigma_u(G)$. 
This proves the first statement in 2, while the second statement follows directly from \cite[Theorem 3.9(2)]{mon1}.
\end{proof}
 
The following result is immediate from Theorem \ref{sc}(2).

\begin{cor}
\label{sccor}
If $0\notin\sigma(G)$, then for any vertex $u$ of $G$, $(0,u)$ and $(1,u)$ are strongly cospectral in $\up{2}G$, and they cannot be strongly cospectral to any other vertex of $\up{2}G$.
\end{cor} 

\begin{exm}
\label{tree}
Let $T$ be a tree with a (unique) perfect matching. Then $T$ has an invertible adjacency matrix, and so each vertex of $\up{2}G$ pairs up with a unique vertex to exhibit strong cospectrality by Corollary \ref{sccor}.
\end{exm}

Note that Example \ref{tree} applies to $P_{2m}$. For $P_{2m+1}$, we have the following observation.

\begin{exm}
\label{tree1}
Let $G=P_{2m+1}$ with edges $(u,u+1)$ for each $u\in\{1,\ldots,2m\}$. Note that 0 is a simple eigenvalue of $G$ with eigenvector $\textbf{e}_1-\textbf{e}_3+\textbf{e}_5-\ldots+(-1)^{\frac{2m-1}{2}}\textbf{e}_{2m+1}$. Therefore, $0\notin \sigma_u(G)$ if and only if $u$ is even. By Theorem \ref{sc}(2), $(0,u)$ and $(1,u)$ are strongly cospectral in $\up{2}G$ if and only if $u$ is even.
\end{exm} 

\section{Pretty good state transfer}
Since strong cospectrality is a necessary condition for PGST, we restrict to the case $n=2$ and $0\not\in\sigma_u(G)$. The following result characterizes PGST in blow-up graphs.

\begin{thm}\label{th1}
Let $u$ be a vertex in $G$ with $0\notin\sigma_u(G).$ The following are equivalent.
\begin{enumerate}
    \item The graph $\up{2}G$ exhibits PGST between $(0,u)$ and $(1,u)$ with phase factor $\gamma=-1.$
    \item There exists a sequence $\tau_k$ so that $\lim\limits_{k\to\infty}\exp{\left(-2i\lambda \tau_k\right)}= -1 \text{ for each } \lambda\in\sigma_u(G)$.
\item $G$ is almost periodic at $u$ with the phase factor $\gamma=-1.$
\item If $m_j$ are integers such that $\hspace{-0.1in}\displaystyle\sum_{\lambda_j\in\sigma_u(G)}\hspace{-0.1in}m_j\lambda_j=0$, then $\hspace{-0.1in} \displaystyle\sum_{\lambda_j\in\sigma_u(G)}\hspace{-0.1in} m_j$ is even.
\end{enumerate}
\end{thm}

\begin{proof}
    We first prove that 1 implies 2. Suppose $\up{2}G$ exhibits PGST between the twin vertices $(0,u)$ and $(1,u)$. We have from  \eqref{eq3}
\begin{equation}
    \label{gam}
    \lim\limits_{k\to 0}\e^T_{(0,u)}U_n(t_k)\e_{(1,u)}=\lim\limits_{k\to 0}\sum\limits_{j=1}^{d}(1/2)\tb{\exp{\left(-2i\lambda_j t_k\right)}-1}(E_j)_{u,u}=\gamma.
\end{equation}
Note that if $(0,u)$ and $(1,u)$ are strongly cospectral in $\up{2}G$, then $\lambda=0$ is the only eigenvalue in $\sigma_{(0,u)}(\up{2}G)$ for which $E_{\lambda}(\be_0\otimes \be_u)=-E_{\lambda}(\be_1\otimes \be_u)$. Thus, $\gamma=-e^{it0}=-1$, and the conclusion in 2 follows from (\ref{gam}). Now, if 2 holds, then there exists a sequence $\tau_k$ so that $\lim\limits_{k\to\infty}\exp{\left(-2i\lambda_j \tau_k\right)}=-1$ for all $\lambda_j\in\sigma_u(G)$. Then \eqref{eq3} yields $\lim\limits_{k\to\infty}\e^T_{(0,u)}U_n(\tau_k)\e_{(1,u)}=-1$, and so 1 holds. Thus, 1 and 2 are equivalent. The equivalence of 2 and 3 follows from the spectral decomposition of $U_2(t)$, while the equivalence of 1 and 4 follows from \cite[Theorem 13]{kirk2}.  
\end{proof}

Denote the largest power of two that divides an integer $a$ by $\nu_2(a)$. The following is a characterization PST in blow-up graphs, the proof of which is identical to that of Theorem \ref{th1}, except that the equivalence of 1 and 4 is a consequence of \cite[Theorem 10]{kirk2}.

\begin{thm}\label{th11}
Let $u$ be a vertex in $G$ with $0\notin\sigma_u(G).$ The following are equivalent.
\begin{enumerate}
    \item The graph $\up{2}G$ exhibits PST between $(0,u)$ and $(1,u)$ with the phase factor $\gamma=-1$.
    \item There exists $\tau\in\Rl$ so that $\exp{\left(-2i\lambda \tau\right)}= -1 \text{ for each } \lambda\in\sigma_u(G)$.
    \item $G$ is periodic at $u$ with the phase factor $\gamma=-1.$
    \item Each $\lambda_j\in\sigma_u(G)$ can be written as $\lambda_j=b_j\sqrt{\Delta}$, where $b_j$ is an integer, $\Delta=1$ or $\Delta>1$ is a square-free integer, and the $\nu_2(b_j)$'s are all equal.
\end{enumerate}
\end{thm} 
We note that while PGST is a relaxation of PST, these two are equivalent for periodic vertices \cite{pal7}. Hence, we say that \textit{proper} pretty good state transfer occurs between two vertices if they exhibit pretty good state transfer but not periodicity. %That is, vertices that admit proper PGST do not admit PST.

\section{Periodicity}

Before we move on to the next section, we include results on periodicity of vertices in blow-up graphs. Similar to strong cospectrality, periodicity is a requirement for two vertices to admit PST. The following is straightforward from \cite[Theorem 6.1]{god2}.

\begin{thm}
\label{per}
Let $n\in\Nl.$ Then $\up{n}G$ is periodic at all vertices in $S_u=\cb{(j,u):j\in\Zl_n}$ if and only if the elements in $\sigma_u(G)$ are either all integers or all integer multiples of $\sqrt{\Delta}$, where $\Delta>1$ is a square-free integer. Moreover, periodicity of vertices in $S_u$ in $\up{n}G$ implies periodicity of $u$ in $G$, and the time at which vertices in $S_u$ are periodic in $\up{n}G$ is equal to $\frac{\tau}{n}$, where $\tau$ is the time at which $u$ is periodic in $G$. Furthermore, if a vertex in $S_u$ is periodic in $\up{n}G$ for some $n,$ then it is periodic in $\up{n}G$ for all $n\in\Nl.$
\end{thm}

From Theorem \ref{per}, periodicity of vertices in $S_u$ in $\up{n}G$ implies periodicity of $u$ in $G$. However, the converse of this does not hold whenever $u$ is periodic in $G$ and $\frac{1}{2}(a+b\sqrt{\Delta})\in\sigma_u(G)$ for some integer $a\neq 0$ and square-free integer $\Delta>1$. Theorem \ref{per} also implies that the blow-up operation yields shorter times for periodicity than the original graph. We illustrate these two observations in the following example.

\begin{exm}
\label{cone}

Let $H$ is a $k$-regular graph on $m$ vertices. Let $G$ be a graph resulting from $H$ by adding a vertex $u$ adjacent to all vertices of $H$. Then $\sigma_{u}(G)=\{\lambda_{\pm}\}$, where $\lambda_{\pm}=\frac{1}{2}(k\pm\sqrt{k^2+4m})$. Thus, $u$ is periodic in $G$ at $\tau=\frac{2\pi}{\sqrt{k^2+4m}}$. Since $0\notin\sigma_u(G)$, Theorem \ref{per}(2) implies that $\up{n}G$ is periodic at all vertices in $S_u=\cb{(j,u):j\in\Zl_n}$ if and only if either $k^2+4m$ is a perfect square , or $k=0$ (in which case $G=K_{1,m}$ and $\lambda_{\pm}=\pm\sqrt{m}$). In both cases, the vertices in $S_u$ are periodic at $\tau=\frac{2\pi}{n\sqrt{k^2+4m}}$. Moreover, if $\up{n}G$ is periodic at all vertices in $S_u=\cb{(j,u):j\in\Zl_n}$ for some $n>1$, then it holds for all $n\in\Nl$.
\end{exm}

In Example \ref{cone}, vertex $u$ in $G$ is always periodic. However, if $k>0$ and $k^2+4m$ is not a perfect square, then each vertex in $S_u$ is not periodic for all $n$. Thus, blow-ups need not preserve periodicity. The next result follows from \cite[Corollary 3.3]{god3}.

\begin{prop}\label{p4}
If $G$ is periodic at $\tau,$ then $\up{n}G$ is periodic at $\frac{\tau}{n}$ for all $n\in\Nl$. 
\end{prop}  

\section{Complete graphs and cycles}\label{qstgs}
We now start our investigation of PST and PGST on blow-ups of various classes of graphs.

Let $K_m$ be a complete graph on $m$ vertices with eigenvalues $-1$ and $m-1,$. If $m$ is even then the eigenvalues $K_{m}$ are all odd. If $\tau=\frac{\pi}{2}$ in Theorem \ref{th11}, then $\up{2}K_{m}$ exhibits PST between every pair of twin vertices. If $m$ is odd, then $\exp{\left(-2i(-1) \tau\right)}=-1$ implies $\exp{\left(-2i(m-1) \tau\right)}= \left[\exp{\left(-2i(-1) \tau\right)}\right]^{1-m}= 1$, and so $\up{2}K_{m}$ does not exhibit PST by Theorem \ref{th11}. Thus, we have:
\begin{thm}
\label{cg}
$\up{2}K_{m}$ exhibits PST at $\frac{\pi}{2}$ if and only if $m$ is even. 
\end{thm}

Next, let $C_n$ be a cycle of length $n$ where $n \geq 3$. The eigenvalues of $C_n$ are
\begin{center}
    $\lambda_l=2\cos{\left(\frac{2l\pi}{n}\right)},~l=0,1,\ldots,n-1.$
\end{center}
We now prove the following negative result for cycles.

\begin{thm}
$\up{2}C_{n}$ does not exhibit PGST for all $n\geq 3$.
\end{thm}
\begin{proof}
If $n$ is divisible by $4,$ then support of each vertex in $C_n$ contains $\lambda=0.$ Thus, $\up{2}C_{n}$ does not admit strong cospectrality by Theorem \ref{sc}. Now, let $n=mp$, where $p$ is an odd prime, and $\omega_p$ is the primitive $p$-th root of unity. Then
%Hence $\up{2}C_{n}$ does not exhibit PGST whenever $n$ is divisible by $4.$ 
\[1+\omega_p+\omega_p^2+\ldots+\omega_p^{p-1}=0.\]
It follows that
\begin{center}
    $1+2\sum_{r=1}^{(p-1)/2}\cos{\left(\frac{2r\pi}{p}\right)}=0.$
\end{center}
Multiplying both sides by $2\cos{\left(2\pi/n\right)}$ yields
\[\lambda_1+\sum_{r=1}^{(p-1)/2}\lambda_{mr+1}+\sum_{r=1}^{(p-1)/2}\lambda_{mr-1}=0.\]
Suppose the sum of the eigenvalues of $C_n$ in the above expression be $L.$ By Theorem \ref{th1}, if there is PGST in $\up{2}C_{n}$ then $1=\lim\limits_{k\to\infty}\exp{\left(-2iL\tau_k\right)}=(-1)^p=-1,$ a contradiction. %Hence the result follows.
\end{proof}

Note that $C_4$ admits PST but $\up{2}C_{4}$ does not. Thus, a blow-up of a graph with PST need not have PST.

\section{Paths}

Next, we characterize PGST in blow-ups of paths. The eigenvalues of $P_n$ are
\begin{center}    $\theta_j=2\cos\left(\frac{j\pi}{n+1}\right),\quad 1\leq j\leq n.$
\end{center}
Moreover, for every vertex $u$ of $P_n$, \cite[Lemma 3.3.6]{Bommel2019} yields
\begin{equation}
\label{supp}
    \sigma_u(P_n)=\{\theta_j:n+1\nmid uj\}
\end{equation}

The following result characterizes PGST on blow-ups of $P_n$ whenever $n=p-1$ or $n=2p-1$ for some prime $p$, or $n=2^t-1$ for some integer $t\geq 2$.

\begin{thm}
\label{pnthm}
  Suppose $n=p-1$ or $n=2p-1$ for some prime $p$, or $n=2^m-1$. Then PGST occurs between vertices $(0,u)$ and $(1,u)$ in $\up{2}P_{n}$ if and only if 
    \begin{enumerate}
        \item $n=p-1$ for some prime $p$ and $u$ is any vertex in $P_n$, or
        \item $n=2p-1$ for some prime $p$ or $n=2^t-1$ for some integer $t$, and $u$ is even.
    \end{enumerate}
\end{thm}

\begin{proof}
    Suppose that either (i) $n=p-1$ for some prime $p$, (ii) $n=2p-1$ for some prime $p$, (iii) $n=2^t-1$ for some integer $t$. From \cite{god4}, (i), (ii) or (iii) hold if and only if the positive eigenvalues of $P_n$ are linearly independent over $\mathbb{Q}$. Thus, if (i) holds, then applying Theorem \ref{kron} with $\alpha_j=\frac{1}{2}(\theta_j+1)$, we obtain integers $q_k$ and $p_j$ such that $|q_k\theta_j-p_j-\alpha_j|<\frac{1}{2\pi k}.$
    Equivalently, $|(2 q_k-1)\frac{\pi}{2}(2\theta_j)-(2 p_j+1)\pi|<\frac{1}{ k}$. Setting $\tau_k=(2 q_k-1)\frac{\pi}{2}$, we obtain
    \begin{equation}
        \label{p-1}\lim\limits_{k\to\infty}\exp{\left(-2i\theta_j\tau_k\right)}=\exp{((2p_j+1)\pi)}=-1
    \end{equation}
    for each eigenvalue $\theta_j\in\sigma_u(P_n)$. Since $P_n$ in this case has an invertible adjacency matrix, $(0,u)$ and $(1,u)$ are strongly cospectral for each vertex $u$ of $P_n$ by Example \ref{tree}. Combining this with (\ref{p-1}) yields PGST between $(0,u)$ and $(1,u)$ for each vertex $u$ of $P_n$ by virtue of Theorem \ref{th1}. This proves 1. The same argument works for cases (ii) and (iii), except that we need $u$ to be an even vertex of $P_n$ by Example \ref{tree1} to ensure that $(0,u)$ and $(1,u)$ are strongly cospectral. Hence, 2 holds.
\end{proof}

Next, we rule out PGST in blow-ups of $P_n$ whenever  $n\notin\{p-1,2p-1,2^t-1\}$.

\begin{thm}
\label{pnlem}
    Let $n=2^tr-1$, where $t\geq 0$ is an integer and $r>0$ is odd. Then PGST does not occur between vertices $(0,u)$ and $(1,u)$ in $\up{2}P_{n}$ whenever
    \begin{enumerate}
        \item $t\geq 0$ and $r$ is an odd composite number, or
        \item $t>0$, $r$ is an odd prime number and $u$ is not a multiple of $2^{t-1}$.
    \end{enumerate}
\end{thm}

\begin{proof}
    To prove 1, let us first suppose that $t=0$ so that $n=r-1$, where $r$ is an odd composite number. From \cite[Equation 26]{bom}, we obtain
    \begin{equation*}
        \sum_{i=0}^{r/p-1}(-1)^i\theta_{c+ip}=0
    \end{equation*}
    where $c\in \{1,2\}$ and $p$ is a prime factor of $r$ such that $p\mid \frac{r}{\operatorname{gcd}(a,r)}$. Since $c+ip$ is not a multiple of $p$ for each $i\in\{0,1,\ldots,\frac{r}{p}-1\}$, (\ref{supp}) implies that each $\theta_{c+ip}$ above is contained in $\sigma_u(P_n)$ for any vertex $u$ of $P_n$. If there is PGST in $\up{2}P_{n}$, then Theorem \ref{th1} yields $\lim\limits_{k\to\infty}\exp{\left(\pm 2i\theta_{c+ip}\tau_k\right)}=-1$. Now, let $L$ be the expression on the left of the above equation. Since $r$ is composite and $p$ is prime, $\frac{r}{p}$ is odd, and so $1=\lim\limits_{k\to\infty}\exp{\left(-2iL\tau_k\right)}=(-1)^{\frac{r}{p}}=-1,$
a contradiction. In case $t>0$, then we have two subcases. First, if $r$ has no prime factor $p$ such that $p\mid \frac{2^tr}{\operatorname{gcd}(a,2^tr)}$, then $r\mid a$ and \cite[Equation 24]{bom} gives us
\begin{equation*}
        \sum_{i=0}^{r/p-1}(-1)^i\theta_{c+i2^tp}=0.
    \end{equation*}
However, if $r$ has a prime factor $p$ such that $p\mid \frac{2^tr}{\operatorname{gcd}(a,2^tr)}$, then \cite[Equation 25]{bom} yields
\begin{equation*}
        \sum_{i=0}^{r/p-1}(-1)^i\theta_{c+i2^t}=0.
    \end{equation*}
In both subcases, an argument similar to the case $t=0$ yields the desired result. Finally, the proof of the second case follows from the first case with $t>0$ and $r$ has no prime factor $p$ such that $p\mid \frac{2^tr}{\operatorname{gcd}(a,2^tr)}$.
\end{proof}

Combining Theorems \ref{pnthm} and \ref{pnlem}, the only case left to check for PGST in blow-ups of $P_n$ is when $n=2^tr-1$, where $t\geq 2$, $r$ is an odd prime and $u$ is a multiple of $2^{t-1}$. For $P_{11}$ with an even vertex $u$, one checks that $\theta_5,\theta_9,\theta_{11}\in\sigma_u(G)$ and $\lambda_5-\theta_9+\theta_{11}=0$. If PGST occurs between $(0,u)$ and $(1,u)$ in $\up{2}P_{11}$, then Theorem \ref{th1} yields $\lim\limits_{k\to\infty}\exp{\left(\pm 2i\theta_{j}\tau_k\right)}=-1$ for each $\theta_j\in\sigma_u(G)$. Consequently, $1=\lim\limits_{k\to\infty}\exp{\left(-2i(\lambda_5-\theta_9+\theta_{11})\tau_k\right)}=(-1)^{3}=-1$,
a contradiction. Thus, $(0,u)$ is not involved in PGST in $\up{2}P_{11}$ for each even vertex $u$ of $P_{11}$. This leads us to the following conjecture.

\begin{con}
\label{con}
Let $n=2^tr-1$, where $t\geq 2$ and $r$ is an odd prime number. If $u$ is a multiple of $2^{t-1}$, then PGST does not occur between $(0,u)$ and $(1,u)$ in $\up{2}P_{n}$.
\end{con}

\begin{rem}
\label{rem3}
In \cite[Theorem 4(2)]{bom}, PGST occurs between $u$ and $12-u$ in $P_{11}$ whenever $u$ is even. As PGST does not occur between $(0,u)$ and $(1,u)$ in $\up{2}P_{n}$,  blow-up graphs need not preserve PGST.
\end{rem}

We end this section with a result on periodicity and PST on blow-ups of $P_n.$ Theorem \ref{per} implies that periodicity of a vertex $(j,u)$ in $\up{m}P_{n}$ implies periodicity of $u$ in $P_n$. From the proof of \cite[Theorem 3.3.8]{Bommel2019}, one can infer that the only periodic vertices in $P_n$ are the vertices of $P_2$ and $P_3$ and the middle vertex $u$ of $P_5$ with $\sigma_u(P_5)=\{0,\pm\sqrt{3}\}$. Applying Theorems \ref{sc}(2) and \ref{th11} yields the following result.

\begin{thm}
Consider the set $S_u=\cb{(j,u):j\in\Zl_m}$. The following hold.
\begin{enumerate}
\item Each vertex in $S_u$ is periodic in $\up{m}P_{n}$ for all $m$ if and only if either (i) $u$ is any vertex of $P_2$, (ii) $u$ is any vertex of $P_3$, or (iii) $u$ is the middle vertex of $P_5$. Moreover, a vertex in $S_u$ has period either $\tau=\frac{2\pi}{m}$, $\tau=\frac{2\pi}{\sqrt{2}m}$, or $\tau=\frac{2\pi}{\sqrt{3}m}$, respectively. 

\item $\up{2}P_{n}$ admits PST between $(0,u)$ and $(1,u)$ if and only if (i) $u$ is any vertex of $P_2$, or (ii) $u$ is the middle vertex of $P_3$. Moreover, the minimum PST time in $\up{2}P_{2}$ is $\tau=\frac{\pi}{2}$. Otherwise, it is $\tau=\frac{\pi}{2\sqrt{2}}$. \end{enumerate}
\end{thm}

\section{Double and subdivided stars}

A double star graph $S_{k,\ell}$ is a tree resulting from attaching $k$ and $\ell$ pendent vertices to the vertices of $K_2$. The following result characterizes PGST in blow-ups of double stars.

\begin{thm}
    Let $u$ and $v$ be the degree $k$ and $\ell$ vertices of $S_{k,\ell}$. The following hold in $\up{2}S_{k,\ell}$.
    \begin{enumerate}
        \item If $k\geq 2$, then $(0,w)$ is not involved in PGST for every $w\in N_{S_{k,\ell}}(u)\backslash\{v\}$.
        \item Let $k=1$. Then PGST occurs between the pairs $\{(0,u),(1,u)\}$, $\{(0,v),(1,v)\}$ and  $\{(0,w),(1,w)\}$, where $w\in N_{S_{k,\ell}}(u)\backslash\{v\}$. Further, if $x\in N_{S_{k,\ell}}(v)\backslash\{u\}$, then PGST occurs between $(0,x)$ and $(1,x)$ if and only if $\ell=1$.
        \item Let $k=\ell$. There is PGST between $(0,u)$ and $(1,u)$ if and only if $4k+1$ is not a perfect square. 
        \item If $k\neq \ell$ and all non-zero eigenvalues of $S_{k,\ell}$ are linearly independent over $\mathbb{Q}$, then PGST occurs between $(0,u)$ and $(1,u)$.
    \end{enumerate}
\end{thm}

\begin{proof}
%Let $A$ be the adjacency matrix of $S_{k,\ell}$.
We first show 1. Since $k\geq 2$, two vertices $w,x\in N_{S_{k,\ell}}(u)$ are twins, and so $T_w\cup T_x$ is a twin set in $\up{2}S_{k,\ell}$ with size at least four by Proposition \ref{twins}. By \cite[Corollary 3.10]{mon1}, $(0,w)$ cannot exhibit strong cospectrality in $\up{2}S_{k,\ell}$ for each $w\in N_{S_{k,\ell}}(u)\backslash\{v\}$. Thus, 1 holds.

To prove 2, let $k=1$ and $w$ be the lone vertex in $N_{S_{1,\ell}}(u)\backslash\{v\}$. Note that $\sigma_u(G)=\sigma_v(G)=\sigma_w(G)=\{\sqrt{(k+2\pm m)/2},-\sqrt{(k+2\pm m)/2}\}$, where $m=\sqrt{k^2+4}$. As $\{\sqrt{(k+2+m)/2},\sqrt{(k+2-m)/2}\}$ is a linearly independent set over $\mathbb{Q}$, the same argument in the proof of Theorem \ref{pnthm}(1) proves the first statement of 2, while 1 and the fact that the positive eigenvalues of $P_4$ yields the second statement of $2$.

To prove 3, suppose $k=\ell$. Then $\sigma_u(G)=\{\frac{1}{2}(1\pm\sqrt{4k+1}),-\frac{1}{2}(1\pm\sqrt{4k+1})\}$. If $4k+1$ is a perfect square, then $1+\sqrt{4k+1}$ and $-1+\sqrt{4k+1}$ are even integers whose largest powers of two dividing them are not equal. Invoking Theorem \ref{th11}(5), PST does not occur between $(0,u)$ and $(1,u)$ in $\up{2}S_{k,\ell}$. Now, if $4k+1$ is not a perfect square, then $\{\pm 1+\sqrt{4k+1}\}$ is a linearly independent set over $\mathbb{Q}$, from which it again follows that PGST occurs between $(0,u)$ and $(1,u)$ in $\up{2}S_{k,\ell}$. The same conclusion can be drawn with the assumption in 4. Hence, 3 and 4 hold. 
\end{proof}

\begin{figure}[!ht]
\centering
\begin{tikzpicture}[scale=.2,auto=left, rotate=0]
\tikzstyle{every node}=[circle, scale=0.65]

 \node[draw] (1) at (0,0) {$0$};
 \node[draw] (2) at (8,0) {$\pm(1-m)$};
 \node[draw] (3) at (16,0) {$1-m$};
 \node[draw] (4) at (-4.5,3.5) {$\pm1$};
 \node[draw] (5) at (-8.5,6.5) {$1$};
 \node[draw] (6) at (-4.5,-3.5) {$\pm1$};
 \node[draw] (7) at (-8.5,-6.5) {$1$};

  \draw[thick,black!70] (3)--(2)--(1)--(4)--(5);
  \draw[thick,black!70] (1)--(6)--(7);
  \draw[dashed,thick,black!70,bend right] (5)edge["$m-1$"](7);
 \end{tikzpicture}
\caption{An eigenvector of $SK_{1,m}$ corresponding to $\pm1.$}
\label{fg3}
\end{figure}
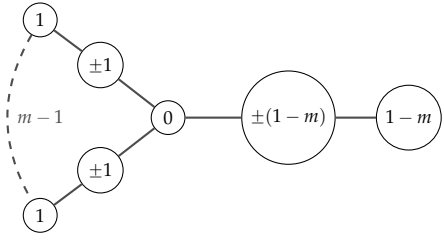

Next, we examine subdivided stars. First, observe that an eigenvector of a graph $G$ can be viewed a function $\x:V(G)\to\Rl,$ where $\x(u)$ is the $u$-th component of the vector $\x.$ Hence $A\x=\lambda\x$ if and only if 
\[\sum\limits_{v\sim u} \x(v)=\lambda \x(u),\text{ for all } u\in V(G).\]
A \textit{subdivided star} $SK_{1,m}$ is obtained by subdividing all edges of the star $K_{1,m}.$ The eigenvalues of $SK_{1,m}$ are known \cite{brou} and we state them below.

\begin{prop}\label{p6}
The eigenvalues of $SK_{1,m}$ for $m>1$ are $\sqrt{m+1},~(1)^{m-1},~0,~(-1)^{m-1},~-\sqrt{m+1}.$
\end{prop}

\begin{proof} The eigenvectors of $SK_{1,m}$ corresponding to $\pm\sqrt{m+1}$ are obtained by assigning the values $1,~\pm\sqrt{m+1}$ and $m$ to all vertices of degree equal to $1,~2$ and $m$, respectively. Similarly, an eigenvector corresponding to $0$ is obtained by assigning the values $1,0$ and $-1$ to all vertices of degree equal to $1,~2$ and $m$, respectively. Let $P$ be the permutation matrix corresponding to an automorphism of $SK_{1,m}.$ If $A$ is the adjacency matrix of $SK_{1,m}$ then $AP=PA.$ Thus, if $\x$ is an eigenvector $SK_{1,m}$ with $A\x=\lambda\x,$ then $P\x$ is also an eigenvector corresponding to $\lambda$ as $AP\x=PA\x=\lambda P\x.$ Thus, we get $m-1$ linearly independent eigenvectors corresponding to $\pm1$ as shown in Figure \ref{fg3}. This yields the spectra of $SK_{1,m}$. 
\end{proof}

Thus, $SK_{1,m}$ is integral if and only if $\sqrt{m+1}$ is an integer. This yields the following:

\begin{cor}
The graph $SK_{1,m}$ is periodic if and only if $\sqrt{m+1}$ is an integer.
\end{cor}
Now we find the support of the vertex of degree $m$ in $SK_{1,m}$ as follows.
\begin{lem}\label{sslm1}
The eigenvalue support of a vertex $u$ of degree $m$ in $SK_{1,m}$ is given by
\[\sigma_u\left(SK_{1,m}\right)=\cb{\sqrt{m+1},~0,~-\sqrt{m+1}}.\]
\end{lem}

\begin{proof} The eigenvectors of $SK_{1,m}$ corresponding to $-\sqrt{m+1}, \sqrt{m+1}$ and $0$, as given in Proposition \ref{p6}. Thus, $-\sqrt{m+1},\sqrt{m+1}$ and $0$ are contained in the eigenvalue support of $u.$ We claim that $\pm1\not\in\sigma_u\left(SK_{1,m}\right).$ It is enough to show that if $\x_{\pm}$ is an eigenvector corresponding to $\pm1$ then $\e_u^T\x_{\pm}= 0.$ Let $v$ be a neighbour of $u$ with $\e_v^T\x_{\pm}\neq 0.$ If $w$ is the leaf adjacent to $v,$ then $\pm1\cdot\e_w^T\x_{\pm}=\e_v^T\x_{\pm} $ and $\pm1\cdot\e_v^T\x_{\pm}=\e_w^T\x_{\pm}+\e_u^T\x_{\pm}.$ Thus $\e_u^T\x_{\pm}=0$, as desired.
\end{proof}

Lemma \ref{sslm1} and Proposition \ref{pp1} combined with Theorem \ref{per} yields the following result.   
\begin{cor}
$SK_{1,m}$ is periodic at the vertex of degree $m$ for all $m>0.$ Moreover, for all $n>0$, the graph $\up{n}SK_{1,m}$ with $m>1$ is periodic at each vertex of degree $nm.$
\end{cor}

We now include the following characterization on PGST in subdivided star.

\begin{thm}
The graph $\up{2}SK_{1,m}$ admits PGST if and only if $\sqrt{m+1}$ is not an even integer. Moreover, if there is PGST on $\up{2}SK_{1,m}$ then it occurs between every pair of twin vertices of degree $4.$ 
\end{thm}

\begin{proof}
    We begin with the observation on the proof of Proposition \ref{p6} that the support of a vertex of degree $1$ and $m$ in $SK_{1,m}$ contains the eigenvalue $0.$ Using Theorem \ref{th1}, we find that there is no PGST in $\up{2}SK_{1,m}$ from either the vertices of degree $2$ or $2m.$ 
    However, the support of a vertex of degree $2$ in $SK_{1,m}$ contain the eigenvalues $\pm \sqrt{m+1}$ and $\pm 1$ only. Now consider the following cases.\\
\textbf{Case I:} Let $\sqrt{m+1}$ be even. If $\up{2}SK_{1,m}$ exhibits PGST then $\lim\limits_{k\to\infty}\exp{\left(-2i(\pm 1)\tau_k\right)}= -1.$ However
\[\lim\limits_{k\to\infty}\exp{\left(-2i\left(\pm \sqrt{m+1}\right)\tau_k\right)}=\left[\lim\limits_{k\to\infty}\exp{\left(-2i(\pm 1)\tau_k\right)}\right]^{\sqrt{m+1}}= 1,\]
a contradiction to Theorem \ref{th1}. In this case, $\up{2}SK_{1,m}$ does not have PGST.\\
\textbf{Case II:} Suppose $\sqrt{m+1}$ is an odd integer. Considering $\tau=\frac{\pi}{2}$, we obtain 
\[\exp{\left(-2i(\pm 1)\tau\right)}= -1, \text{ and } \exp{\left(-2i\left(\pm \sqrt{m+1}\right)\tau\right)}= -1.\]
By Theorem \ref{th11}, $\up{2}SK_{1,m}$ admits PST between every pair of twins of degree $4$.\\
\textbf{Case III:} If $m+1$ is not a perfect square then $1,\sqrt{m+1}$ are linearly independent over $\Ql.$ Assuming $\alpha=\frac{\sqrt{m+1}}{2}+\frac{1}{2}$ in Theorem \ref{kron}, for every $k\in\Nl$, there exists integers $p_k$ and $q_k$ such that
   \begin{center}
       $\left| q_k\sqrt{m+1} - p_k - \frac{\sqrt{m+1}}{2}-\frac{1}{2}\right|<\frac{1}{2\pi k}.$
   \end{center}
    This further yields
    \begin{center}
        $\left| \left(2q_k-1\right)\frac{\pi}{2}\cdot 2\sqrt{m+1} - \left(2p_k+1\right)\pi \right|<\frac{1}{k}.$
    \end{center}
 Considering $\tau_k=\left(2q_k-1\right) \frac{\pi}{2},$ we find that \[\lim\limits_{k\to\infty}\exp{\left(-2i(\pm 1)\tau_k\right)}= -1, \text{ and } \lim\limits_{k\to\infty}\exp{\left(-2i\left(\pm\sqrt{m+1}\right)\tau_k\right)}= -1.\]
By Theorem \ref{th1}, we get that $\up{2}SK_{1,m}$ admits PGST between every pair of twins of degree $4$.% whenever $m+1$ is not a perfect square.
\end{proof}

\section{Cones}\label{cones}
A \textit{cone} $G$ on a graph $H$ is a graph with vertex set $V(G)=\{u\}\cup V(H)$ and $G$ is obtained by adding a vertex $u\notin V(H)$ that is adjacent to all vertices of $H$. We call $u$ the \textit{apex} of $G$. It can be checked that if $H$ is a $k$-regular graph on $m$ vertices and $G$ is a cone $K_1+H$, then $\up{2}G$ is a (disconnected) \textit{double cone} on $\up{2}H$, where $\up{2}H$ is a $2k$-regular graph on $2m$ vertices. That is, $\up{2}G$ is a graph obtained by adding two vertices $u$ and $w$ that are adjacent to all vertices of $\up{2}H$. We call $u$ and $w$ \textit{apexes} of the double cone. Periodicity, PST and PGST between apexes of disconnected double cones over regular graphs were characterized by Kirkland et al.\ \cite{kirk2}. The following is immediate from Theorems 11(1) and 14(2a) in \cite{kirk2}.

\begin{thm}\label{thc1}
Let $G$ be a cone on a $k$-regular graph on $m$ vertices with apex $u$. %The following hold in $\up{2}G$.

\begin{enumerate}
    \item PST occurs between $\left(0,u\right)$ and $\left(1,u\right)$ in $\up{2}G$ if and only if either (i) $k=0$ or (ii) $k>0$ and $m=\frac{1}{4}s\left(2k+s\right)$ for some even integer $s$ such that $\nu_2(2k)>\nu_2(s)$. Moreover, PST occurs at time $\tau=\frac{\pi}{2\sqrt{m}}$ whenever $k=0$ and $\tau=\pi/g$ otherwise, where $g=\operatorname{gcd}(2k,s)$.
    \item Proper PGST occurs between $\left(0,u\right)$ and $\left(1,u\right)$ in $\up{2}G$ if and only if $k>0$ and $k^2+4m$ is not a perfect square.
\end{enumerate}
\end{thm}

We illustrate the above theorem using the following example.

\begin{exm}
\label{4}
    Let $G$ be a cone on $H$ with apex $u$, where $H$ is $k$-regular with $m$ vertices.
\begin{enumerate}
\item If $H$ is the empty graph on $m$ vertices, then $k=0$ and $G=K_{1,m}$. By Theorem \ref{thc1}(1), PST occurs between $\left(0,u\right)$ and $\left(1,u\right)$ in $\up{2}G$ at $\frac{\pi}{2\sqrt{m}}$ for all $m\geq 1$. 
\item For all $m\geq 3$, let $H=C_m$, so that $k=2$.
\begin{enumerate}
    \item Let $m+1$ be a perfect square. Since $m=\frac{1}{4}s(2k+s)$, we obtain $s=2(\pm \sqrt{m+1}-1)$. By Theorem \ref{thc1}(1), we get PST between $\left(0,u\right)$ and $\left(1,u\right)$ in $\up{2}G$ if and only if $m$ is odd.
    \item If $m+1$ is not a perfect square, then so is $k^2+4m=4(m+1)$. By Theorem \ref{thc1}(2), proper PGST occurs between $\left(0,u\right)$ and $\left(1,u\right)$ in $\up{2}G$.
\end{enumerate}
\item Let $H=K_{r,r}$ be a complete bipartite graph. Then $k^2+4m=(r-1)^2+4(2r)$ is not a perfect square. By Theorem \ref{thc1}(2), proper PGST occurs between $\left(0,u\right)$ and $\left(1,u\right)$ in $\up{2}G$ for all $r\geq 2$.
\end{enumerate}
\end{exm}

\section{Cartesian products}\label{cons}

PST is rare, and this fact motivates us to search for new graphs with PST. It is desirable for these graphs to be sparse, and for PST to occur between vertices that are far apart. For the case of blow-up graphs, PST in $\up{2}G$ can only occur between vertices at distance two. Nonetheless, we may still utilize our results herein to illustrate the construction of relatively sparse graphs using blow-up graphs, with the property that PST occurs between vertices that are far apart. To do this, we make use of the Cartesian product operation.

Recall that the Cartesian product of graphs $G$ and $H$, denoted $G\square H$, is the graph with adjacency matrix $A(G)\otimes I+I\otimes A(H)$. Thus,
\begin{center}
    $U_{G\square H}(t)=U_G(t)\otimes U_H(t)$
\end{center}
Consequently, PST occurs between $(u,x)$ and $(v,y)$ in $G\square H$ whenever PST occurs between $u$ and $v$ in $G$ and PST occurs between $x$ and $y$ in $H$ at the same time. Combining this fact with \cite[Theorem 4.2]{pal5} yields the following observations.

\begin{exm}
    Suppose $\up{2}G$ has PST between $u$ and $v$ at time $\tau=\frac{\pi}{g}$. Denote the hypercube of dimension $d$ by $Q_d$, which is known to admit PST between  vertices $x$ and $y$ that are at distance $d$. If $g$ is an even integer, then $(\up{2}G)\square Q_d$ has PST between $(u,x)$ and $(v,y)$ for all $d\geq 2$ at time $\frac{\pi}{2}$. However, if $g\notin\mathbb{Q}$, then $(\up{2}G)\square Q_d$ has PGST between $(u,x)$ and $(v,y)$ for all $d\geq 2$.
\end{exm}

\begin{exm}
From Example \ref{4}(1), we know $\up{2}K_{1,m}$ admits PST, say between vertices $u$ and $v$, at time $\frac{\pi}{2\sqrt{m}}$ for all $m\geq 1$. Let $X=X(n)$ be the Cartesian product of $n$ copies of $P_3$, which is known to admit PST between vertices $x$ and $y$ that are at distance $2^n$. If $m=2n^2$, then  $(\up{2}G)\square X$ has PST between $(u,x)$ and $(v,y)$ at time $\frac{\pi}{\sqrt{2}}$. Otherwise, $(\up{2}G)\square X$ has PGST between $(u,x)$ and $(v,y)$.
\end{exm}

\section{Future work}

In this paper, we characterized strong cospectrality, periodicity, PST and PGST in blow-up graphs, and applied our results to the blow-ups of complete graphs, cycles, paths, double stars, subdivided stars and cones. Our results 

One line of investigation would be to determine the quantum state transfer properties of blow-ups of other classes of graphs such as Cayley and trees in general.

We would also like to know whether Conjecture \ref{con} is true. If it is not, then for which integers $t\geq 2$ and odd primes $r$ does $\up{2}P_{2^tr-1}$ exhibit PGST between $(0,u)$ and $(1,u)$, whenever $u$ is a multiple of $2^{t-1}$?

Lastly, if $G$ is a cone on a non-regular graph, then when does $\up{2}G$ admit periodicity and PGST?

\section*{Acknowledgement}
We thank C.\ van Bommel for providing the $P_{11}$ example. % example that lead us to Conjecture \ref{con}. 
Hermie Monterde is supported by the University of Manitoba Faculty of Science and Faculty of Graduate Studies. Hiranmoy Pal is supported by Science and Engineering Research Board (Project No. SRG/2021/000522).

%%%%%%% THE BIBLIOGRAPHY %%%%%%%
\bibliographystyle{unsrt}
\bibliography{edge_perturbation}

\end{document}